\documentclass[a4paper,UKenglish,cleveref, autoref, thm-restate]{lipics-v2021}
 
\usepackage[utf8]{inputenc}

\usepackage{tabularx}

\usepackage{amsthm}
\usepackage{amsmath}
\usepackage{amssymb}
\usepackage{amsfonts}
\usepackage{mathtools}
\usepackage{enumitem}
\usepackage{hyperref}
\usepackage{algorithm}
\usepackage{multicol}
\usepackage[noend]{algpseudocode}
\algnewcommand\algorithmicinput{\textbf{Input:}}
\algnewcommand\algorithmicoutput{\textbf{Output:}}
\algnewcommand\Input{\item[\algorithmicinput]}
\algnewcommand\Output{\item[\algorithmicoutput]}

\pdfoutput=1 %uncomment to ensure pdflatex processing (mandatatory e.g. to submit to arXiv)
\hideLIPIcs  %uncomment to remove references to LIPIcs series (logo, DOI, ...), e.g. when preparing a pre-final version to be uploaded to arXiv or another public repository

\usepackage[sort,numbers]{natbib}

\usepackage{makecell}

\usepackage{tikz}

\usepackage[linecolor=green!70!black, backgroundcolor=green!10, bordercolor=black, textsize=scriptsize, obeyFinal
]{todonotes}

\usepackage{xspace}

\newtheorem{hypothesis}[theorem]{Hypothesis}

\newcommand{\threefield}[3]{$#1\mid#2\mid#3$}
\newcommand{\commentout}[1]{}

\DeclarePairedDelimiterX{\abs}[1]{\lvert}{\rvert}{#1}

\newcommand{\Prob}{\threefield{1}{}{\sum w_jU_j}\xspace}
\newcommand{\SDProb}{\threefield{1}{d_j=d}{\sum w_jU_j}\xspace}
\newcommand{\UWProb}{\threefield{1}{}{\sum U_j}\xspace}
\newcommand{\SPProb}{\threefield{1}{p_j=p}{\sum w_jU_j}\xspace}
\newcommand{\PProb}{\threefield{1}{}{\sum p_jU_j}\xspace}

\usepackage{etoolbox}
\newcommand{\appsymb}{$\star$}
\newcommand{\appref}[1]{{\hyperref[proof:#1]{\appsymb}}}
\newcommand{\appLink}[1]{{\hyperref[#1]{\appsymb}}}

\title{Minimizing the Weighted Number of Tardy Jobs via (max,+)-Convolutions}

\author{Danny~Hermelin}{Department of Industrial Engineering and Management, Ben-Gurion~University~of~the~Negev, 
Beer-Sheva, 
Israel}{hermelin@bgu.ac.il}{}{}%TODO mandatory, please use full name; only 1 author per \author macro; first two parameters are mandatory, other parameters can be empty. Please provide at least the name of the affiliation and the country. The full address is optional. Use additional curly braces to indicate the correct name splitting when the last name consists of multiple name parts.

\author{Hendrik~Molter}{Department of Industrial Engineering and Management, Ben-Gurion~University~of~the~Negev, 
Beer-Sheva, 
Israel}{molterh@post.bgu.ac.il}{}{}

\author{Dvir~Shabtay}{Department of Industrial Engineering and Management, Ben-Gurion~University~of~the~Negev, 
Beer-Sheva, 
Israel}{dvirs@bgu.ac.il}{}{}

\authorrunning{Danny Hermelin, Hendrik Molter, and Dvir Shabtay} %TODO mandatory. First: Use abbreviated first/middle names. Second (only in severe cases): Use first author plus 'et al.'

\Copyright{Danny Hermelin, Hendrik Molter, and Dvir Shabtay} %TODO mandatory, please use full first names. LIPIcs license is "CC-BY";  http://creativecommons.org/licenses/by/3.0/

%, Parameterized Algorithms} %TODO mandatory; please add comma-separated list of keywords

\category{} %optional, e.g. invited paper

\relatedversion{} %optional, e.g. full version hosted on arXiv, HAL, or other respository/website
\funding{Supported by the ISF, grant No.~1070/20.}%optional, to capture a funding statement, which applies to all authors. Please enter author specific funding statements as fifth argument of the \author macro.

\acknowledgements{}%optional

\nolinenumbers %uncomment to disable line numbering

\EventEditors{John Q. Open and Joan R. Access}
\EventNoEds{2}
\EventLongTitle{42nd Conference on Very Important Topics (CVIT 2016)}
\EventShortTitle{CVIT 2016}
\EventAcronym{CVIT}
\EventYear{2016}
\EventDate{December 24--27, 2016}
\EventLocation{Little Whinging, United Kingdom}
\EventLogo{}
\SeriesVolume{42}
\ArticleNo{23}
\begin{document}
%\nolinenumbers
\maketitle

\begin{abstract}
The \Prob problem asks to determine -- given $n$ jobs each with its own processing time, weight, and due date --  the minimum weighted number of tardy jobs in any single machine non-preemptive schedule for these jobs. This is a classical scheduling problem that generalizes both Knapsack, and Subset Sum. The best known pseudo-polynomial algorithm for \Prob, due to Lawler and Moore [Management Science'69], dates back to the late 60s and has a running time of $O(d_{\max}n)$, where $n$ is the number of jobs and $d_{\max}$ is their maximal due date. A recent lower bound by Cygan \emph{et al.}~[ICALP'19] for Knapsack shows that \Prob cannot be solved in $\widetilde{O}((n+d_{\max})^{2-\varepsilon})$ time, for any $\varepsilon > 0$, under a plausible conjecture. This still leaves a gap between the best known lower bound and upper bound for the problem. 

In this paper we design a new simple algorithm for \Prob that uses $(\max,+)$-convolutions as its main tool, and outperforms the Lawler and Moore algorithm under several parameter ranges. In particular, depending on the specific method of computing $(\max,+)$-convolutions, its running time can be bounded by
\begin{multicols}{2}
\begin{itemize}
\item $\widetilde{O}(n+d_{\#}d_{\max}^2)$.
\item $\widetilde{O}(d_{\#}n +d^2_{\#}d_{\max}w_{\max})$.
\item $\widetilde{O}(d_{\#}n +d_{\#}d_{\max}p_{\max})$.
\item $\widetilde{O}(n^2 +d_{\max}w^2_{\max})$.
\item $\widetilde{O}(n^2 + d_{\#}(d_{\max}w_{\max})^{1.5})$.
\item[\vspace{\fill}]
\end{itemize}
\end{multicols}
\noindent Here, $d_{\#}$ denotes the number of \emph{different} due dates in the instance, $p_{\max}$ denotes the maximum processing time of any job, and $w_{\max}$ denotes the maximum weight of any job. 

To obtain these running  times we adapt previously known $(\max,+)$-convolution algorithms to our setting, most notably the prediction technique by Bateni \emph{et al.}~[STOC'19], the techniques for concave vectors by Axiotis and Tzamos~[ICALP'19], and the techniques for bounded monotone vectors by Chi \emph{et al.}~[STOC'22]. Moreover, to complement the $\widetilde{O}((n+d_{\max})^{2-\varepsilon})$ lower bound of Cygan \emph{et al.}, we show that \Prob is also unlikely to admit an algorithm with running time $\widetilde{O}(n^{O(1)}+d_{\max}^{1-\varepsilon})$, for any $\varepsilon > 0$. 
\end{abstract}

\clearpage

%\linenumbers
\section{Introduction}
\label{sec:intro}%

One of the most fundamental problems in the area of scheduling is the problem of non-preemptively scheduling a set of jobs on a single machine so as to minimize the weighted number of tardy jobs. In this problem, we are given a set of $n$ jobs $J=\{1,\ldots,n\}$, where each job~$j$ has a \emph{processing time} $p_j \in \mathbb{N}$, a \emph{weight} $w_j \in \mathbb{N}$, and a \emph{due date} $d_j \in \mathbb{N}$. A \emph{schedule}~$\sigma$ for $J$ is a permutation $\sigma: \{1,\ldots,n\} \to \{1,\ldots,n\}$ specifying the processing order of  the jobs. In a given schedule~$\sigma$, the \emph{completion time} $C_j$ of a job $j$ under $\sigma$ is  $C_j = \sum_{\sigma(i) \leq \sigma(j)} p_i$; that is, the total processing time of jobs preceding $j$ in $\sigma$ (including $j$ itself). Job $j$ is \emph{tardy} in $\sigma$ if $C_j > d_j$, and \emph{early} otherwise. Our goal is to find a schedule where the total weight of tardy jobs is minimized. Following the standard three field notation of Graham~\cite{Graham1969}, we use \Prob to denote this problem throughout the paper.

The \Prob problem models a very basic and natural scheduling scenario, and is thus very important in practice. However, it also plays a prominent theoretical role, most notably in the theory of scheduling algorithms, being one of the first scheduling problems shown to be NP-hard~\cite{Karp72} and have a fully polynomial time approximation scheme (FPTAS)~\cite{Sahni76}. Moreover, the special case of \Prob where jobs have a single common due date is essentially equivalent to the Knapsack problem, and if furthermore the processing time of each job equals its weight, then the problem becomes Subset Sum~\cite{Karp72}. Thus, \Prob generalizes two of the most basic problems in combinatorial optimization, and so practically any standard textbook on scheduling (\emph{e.g.}~\cite{Brucker2006,Pinedo2008}) devotes some sections to this problem and its variants. We refer the reader also to a relatively recent survey by Adamu and Adewumi~\cite{AdamuAdewumi2014} for further background on the problem.

%Another reason \Prob is such a prominent problem in the area of scheduling is that it is a natural generalization of the classical Knapsack problem. Indeed, the special case where all jobs have a common due date (\emph{i.e.} $d_1= \cdots = d_n=d$), denoted as the \SDProb problem, translates directly to the dual version of Knapsack: If all due dates equal to $d$, our goal is to minimize the total weight of jobs that complete after $d$, where in Knapsack we wish to maximize the total weight of jobs that complete before $d$. (Here $d$ corresponds to the Knapsack size, the processing times correspond to item sizes, and the weights correspond to item values.) This equivalence was already observed in Karp's classical paper~\cite{Karp72}, albeit under a slightly different version of Knapsack. 
%\begin{theorem}[\cite{Karp72}]
%\label{thm:BasicNPHardness}%
%Any algorithm for Knapsack running in $T(n)$ time can be used to solve the \SDProb problem in $O(T(n))$ time, and vice-versa. In particular, \SDProb is (weakly) NP-hard.
%\end{theorem}

As \Prob generalizes both Knapsack and Subset Sum, it is NP-hard, but only \emph{weakly} NP-hard. This means that it admits pseudo-polynomial time algorithms, algorithms with polynomial running time when the input numbers are encoded in unary. The famous Lawler and Moore algorithm~\cite{LawlerMoore} for \Prob dates way back to the late 60s, extending Bellman's classical Knapsack algorithm~\cite{bellman1957dynamic} from the late 50s, is a prime example of a pseudo-polynomial time scheduling algorithm. It has a running time of $O(d_{\max} \cdot n)$, where $d_{\max}=\max_j d_j$ is the maximum due date among all jobs in the instance. 
\begin{theorem}[\cite{LawlerMoore}]
\label{thm:LawlerMoore}%
\Prob can be solved in $O(d_{\max} \cdot n)$.
\end{theorem}

Despite the simplicity of the Lawler and Moore algorithm, and despite the fact that its generic nature also allows solving a multitude of other problems~\cite{LawlerMoore}, it is still the fastest algorithm known for \Prob in the realm of pseudo-polynomial time algorithms. This begs the following natural question:
\begin{quote}
``Can we obtain faster pseudo-polynomial time algorithms for \Prob, or alternatively, prove that these do not exist under some plausible conjecture?"
\end{quote}
As an illustrative example, consider the case where $d_{\max}=O(\sqrt{n})$. Can we obtain a running time better than the $O(n^{1.5})$ given by Lawler and Moore's algorithm? Can we perhaps get linear or near-linear time? Such questions are the motivation behind this paper.  \\

%This question, which as far as we can tell has never been really addressed, is the main motivation behind this paper. \\

%\noindent \textbf{The single due date case:} 

\subsection{The single due date case}

%It turns out that the answer to the above question is positive for the special case of \Prob where all jobs have a common single due date, the \SDProb problem. 

The special case of \Prob where all jobs have a common single due date, is denoted by \SDProb in Graham's notation. As mentioned above, \SDProb translates directly to a ``dual version'' of Knapsack: If all due dates equal to $d$, our goal is to minimize the total weight of jobs that complete after~$d$, where in Knapsack we wish to maximize the total weight of jobs that complete before~$d$. (Here $d$ corresponds to the Knapsack size, the processing times correspond to item sizes, and the weights correspond to item values.) Thus, the two problems are equivalent under exact algorithms. 

In a recent breakthrough result, using an elegant layering technique introduced by Bringmann for Subset Sum~\cite{Bring17}, Cygan~\emph{et al.}~\cite{CyganMWW19} devised a randomized algorithm that solves the \SDProb problem in $\widetilde{O}(n+d_{\max}^2)$ time\footnote{We use $\widetilde{O}$ to suppress polylogarithmic factors.}. (Here, $d_{\max}$ is simply the common due date of all jobs.) 
\begin{theorem}[\cite{CyganMWW19}]
\label{thm:CyganUpperBound}%
There is a randomized algorithm solving \SDProb in $\widetilde{O}(n+d_{\max}^2)$ time.
\end{theorem}
Thus, in the above example where $d_{\max} = O(\sqrt{n})$ the algorithm above gives near-linear time, and in general improves upon the Lawler and Moore algorithm whenever $d_{\max} = \widetilde{o}(n)$.
%The answer to the question we posed above is therefore affirmative in the single due date case. 
Can we get a similar result for the multiple due date case?

\subsection{(max,+)-convolutions}

Following Cygan~\emph{et al.}~\cite{CyganMWW19}, there has been a surge in fast pseudo-polynomial time algorithms for Knapsack~\cite{AxiotisTzamos19,BateniHSS18,BringmannC22,DChiDX022,PolakEtAl21}, each outperforming the other under certain conditions. However, a common thread to all of these, including the Cygan~\emph{et al.} paper, is the heavy reliance on $(\max,+)$-convolutions. These are also central to this paper as well. 
\begin{definition}%[(max,+)-Convolution]
\label{def:MaxPlusConv}%
Let $A = (A[k])^m_{k=0}$ and $B = (B[\ell])^n_{\ell=0}$ be two integer vectors with $m \leq n$. The $(\max,+)$-convolution of $A$ and $B$, denoted $A \oplus B$, is defined as the vector $C= (C[\ell])^{n}_{\ell=0}$ where 
$$
C[\ell] =\max_{0 \leq k \leq \ell} \big( A[k]+B[\ell - k] \big) \text{ for each } \ell \in \{0,\ldots,n\}.
$$
\end{definition}

For general integer vectors of length at most $n$, it is widely believed that one cannot compute their $(\max,+)$-convolution in $O(n^{2-\varepsilon})$-time. This is known as the \emph{$(\min,+)$-conjecture}~\cite{CyganMWW19}. However, this barrier can be broken when the vectors have certain structure. Bateni \emph{et al.}~\cite{BateniHSS18} showed that when one can compute so-called ``range intervals'' for the two vectors, then one can compute their $(\max,+)$-convolution in $\widetilde{o}(n^2)$ time using their \emph{prediction technique}. Axiotis and Tzamos~\cite{AxiotisTzamos19} showed that when one vector is so called \emph{$s$-step concave}, then the $(\max.+)$-convolution can be computed in $O(n)$ time. Finally, Chi \emph{et al.}~\cite{DChiDX022} presented an $\widetilde{O}(n^{1.5})$-time algorithm for the case where both vectors are monotone non-decreasing and have values bounded by $O(n)$.

As previously mentioned, these fast $(\max,+)$-convolution algorithms have been used to devise fast psuedo-polynomial time algorithms for \SDProb (or rather Knapsack). Bateni \emph{et al.}~\cite{BateniHSS18} used their prediction technique to devise an algorithm running in time $\widetilde{O}(n+d_{\max}w_{\max})$, where $w_{\max} = \max_j w_j$ is the maximum weight of the instance. 
%(They also presented an algorithm running in  $\widetilde{O}(np_{\max}+d_{\max}p_{\max})$ time, where $p_{\max}= \max_j p_j$, via an alternative method.) 
Axiotis and Tzamos~\cite{AxiotisTzamos19} used their technique to obtain $\widetilde{O}(n+d_{\max}p_{\max})$ and $\widetilde{O}(nw^2_{\max})$ time algorithms. Bringmann and Cassis~\cite{BringmannC22} used the techniques of Chi \emph{et al.}~\cite{DChiDX022} to provide an $\widetilde{O}((nw_{\max}+d_{\max})^{1.5})$-time algorithm. Lastly, Polak \emph{et al.}~\cite{PolakEtAl21} used the techniques of Axiotis and Tzamos~\cite{AxiotisTzamos19}, along with proximity techniques similar to those found in~\cite{EisenbrandWeismantel20}, to obtain algorithms with running-time $\widetilde{O}(n+w_{\max}^3)$ and $\widetilde{O}(n+p_{\max}^3)$.

All algorithms mentioned above improve upon the algorithm of Cygan~\emph{et al.}~\cite{CyganMWW19} for \SDProb in certain parameter ranges. However, Cygan~\emph{et al.} did not only use $(\max,+)$-convolutions to solve \SDProb, but they also showed a reduction in the other direction. Namely, a fast algorithm for \SDProb would refute the $(\max,+)$-conjecture mentioned above. %This implies that the algorithm of Theorem~\ref{thm:CyganUpperBound} cannot be substantially improved for general \SDProb instances, assuming the $(\max,+)$-conjecture is true.
\begin{theorem}[\cite{CyganMWW19}]
\label{thm:CyganLowerBound}%
There is no $\widetilde{O}((n+d_{\max})^{2-\varepsilon})$-time algorithm for \SDProb, for any $\varepsilon > 0$, assuming the $(\max,+)$-conjecture holds. 
\end{theorem}
Observe that the lower bound above does not exclude a linear or even sublinear running time dependency in $d_{\max}$. In particular, the existence of an $\widetilde{O}(n^2+d^{1/100}_{\max})$ time algorithm is still open. Nevertheless, if we insist on a subquadratic dependency on $n$, the algorithm of Theorem~\ref{thm:CyganUpperBound} is optimal\footnote{As we neglect polylogarithmic factors throughout the paper, we use the term optimal to mean up to polylogarithmic factors, and assuming some plausible conjecture.}. %Moreover note that for our example above of $d_{\max}=O(\sqrt{n})$, this algorithm has $\widetilde{O}(n)$ running time. Can this also be achieved for the multiple due date case?

\subsection{Our results}

Inspired by the recent pseudo-polynomial algorithms for Knapsack, we present in Section~\ref{sec:AlgOverview} a simple algorithm for \Prob that uses $(\max,+)$-convolutions in a direct fashion. In particular, our algorithm first partitions the set of jobs into $d_{\#}$ subsets, where all jobs within the same subset have the same due date. Each subset of jobs can be viewed as a single \SDProb (or Knapsack) instance, so the algorithms can use known algorithms to obtain solution vectors for each of these instances. It then combines all solution vectors into a single vector for the entire \Prob instance using $(\max,+)$-convolutions. The running time of the entire algorithm is $\widetilde{O}(n+d_{\#}d_{\max}^2)$ time.
\begin{theorem}
\label{thm:MainAlgorithm}
\Prob can be solved in $\widetilde{O}(n+d_{\#}d_{\max}^2)$ time.
\end{theorem}

By the lower bound of Theorem~\ref{thm:CyganLowerBound}, the algorithm above is optimal when the job instance has~$\widetilde{O}(1)$ different due dates. Moreover, as $d_{\#} \leq d_{\max}$, its running time can also be bounded by $\widetilde{O}(n+d_{\max}^3)$. Thus, in our running example with $d_{\max}=O(\sqrt{n})$, the running time of this algorithm matches the one by Lawler and Moore, and improves upon it whenever $d_{\#}=O(n^{1/2-\varepsilon})$ for any $\varepsilon > 0$. 

We next explore improvements on our algorithm using recent technique for fast $(\max,+)$-convolutions. In our first result of this flavor, we adapt the prediction framework of Bateni \emph{et al.}~\cite{BateniHSS18} to the multiple due date case. We show that their algorithm for Knapsack can be extended to handle $d_{\#}$ different due dates, at a cost of an increase of a factor of $d^2_{\#}$ to its running time. More specifically, in Section~\ref{sec:prediction} we prove that:
\begin{theorem}
\label{thm:MainSpeedup}%
\Prob can be solved in $\widetilde{O}(d_{\#}n +d^2_{\#}d_{\max}w_{\max})$ time. 
\end{theorem}
Note that the running time above improves upon Theorem~\ref{thm:MainAlgorithm} whenever $d_{\#}w_{\max}=\widetilde{o}(d_{\max})$. In particular, if both $d_{\#}=\widetilde{O}(1)$ and $w_{\max}=\widetilde{O}(1)$, this running time becomes $\widetilde{O}(n +d_{\max})$.  

We then consider two other recent techniques for fast $(\max,+)$-convolutions, namely the technique by Axiotis and Tzamos~\cite{AxiotisTzamos19} and Chi \emph{et al.}~\cite{DChiDX022}
for $s$-step concave and bounded monotone vectors respectively. As opposed to the prediction technique by Bateni \emph{et al.}, these are relatively straightforward to adapt to our setting. Nevertheless, they still provide us with improved running times in certain ranges of parameters $d_{\#}$, $d_{\max}$, $w_{\max}$, and $p_{\max}$. Theorem~\ref{thm:SecondSpeedup} below summarizes all the running times we can obtain this way:
\begin{theorem}
\label{thm:SecondSpeedup}%
\Prob can be solved in either 
\begin{itemize}
\item $\widetilde{O}(d_{\#}n +d_{\#}d_{\max}p_{\max})$ time, 
\item $\widetilde{O}(n^2 +d_{\max}w^2_{\max})$ time, or
\item $\widetilde{O}(n^2 + d_{\#}(d_{\max}w_{\max})^{1.5})$ time.
\end{itemize}
\end{theorem}

Finally, recall that the lower bound of Theorem~\ref{thm:CyganLowerBound} does not exclude algorithms with running times of the form $\widetilde{O}(n^{O(1)}+d_{\max}^{2-\varepsilon})$, for any $0 < \varepsilon < 2$. By exploiting a construction previously used in~\cite{AbboudBHS20}, we show that \Prob is unlikely to admit an $\widetilde{O}(n^{O(1)}+d_{\max}^{1-\varepsilon})$ time algorithm, for any $\varepsilon > 0$. Our lower bound is based on completely different hypothesis than the $(\max,+)$-conjecture, namely the \emph{$\forall \exists$ strong exponential time hypothesis ($\forall \exists$-SETH)}, which is akin to the well known strong exponential time hypothesis of Impagliazzo, Paturi, and Zane~\cite{IP2001,IPZ2001}. 
\begin{theorem}
\label{thm:LowerBound}%
Assuming $\forall \exists$-SETH, there is no algorithm for \Prob running in time $\widetilde{O}(n^{O(1)}+d_{\max}^{1-\varepsilon})$, for any $\varepsilon > 0$.
\end{theorem}
Note that as $d_{\#} \leq n$, we could also replace the first term in the lower bound above with the term~$(d_{\#}n)^{O(1)}$. 

\subsection{Further related work}

As mentioned above, \Prob is one of the earliest studied problems in the field of combinatorial optimization in general, and in scheduling theory in particular~\cite{LawlerMoore}. Karp placed the problem in the pantheon of combinatorial optimization problems by listing it in his landmark 1972 paper~\cite{Karp72}. The problem is known to be polynomial-time solvable in a few special
cases: Moore~\cite{Moore68} provided an $O(n \log n)$ time algorithm for solving the unit
weight \UWProb problem, and Peha~\cite{Peha95} presented an $O(n \log n)$ time algorithm
for \SPProb (the variant where all jobs have equal processing time), see also~\cite{BruckerK06}. Sahni showed that \Prob admits an FPTAS, in yet another landmark paper in the area of scheduling~\cite{Sahni76}. 

%The single due date special case of \Prob, \emphh{i.e.} the Knapsack problem, is one the most widely studied problems in combinatorial optimization, and surveying all the research done on this problem is far beyond the scope of this paper. We briefly mention in passing that 

Exact algorithms for the problem based on a branch-and-bound procedure can be found in~\cite{MHallahBulfin,Tang,VillarrealBulfin}. The problem is known to be polynomial-time solvable when either the number of different processing times $p_{\#}$, or the number of different weights $w_{\#}$, is bounded by a constant, and fixed parameter tractable when parameterized by either $p_{\#}+w_{\#}$, $p_{\#}+d_{\#}$, or $w_{\#}+d_{\#}$~\cite{HermelinKPS21}. Finally, the special case where $w_j=p_j$ for all jobs $j$, the \PProb problem, was recently shown to be solvable in $\widetilde{O}(p^{7/4})$, $\widetilde{O}(nd_{\#}p_{\max})$, and $\widetilde{O}(np_{\max}+d_{\#}d_{\max})$ time~\cite{BringmannFHSW20}.

\section{Algorithm Overview}
\label{sec:AlgOverview}%

In this section we present an overview of a generic algorithm for solving \Prob using $(\max,+)$-convolutions. In particular, we provide a complete proof of Theorem~\ref{thm:MainAlgorithm}. Later on, this algorithm will also be used for proving Theorems~\ref{thm:MainSpeedup} and \ref{thm:SecondSpeedup}.

A key observation, used already in the Lawler and Moore~\cite{LawlerMoore} algorithm, is that any instance of \Prob always has an optimal schedule which is an Earliest Due Date schedule. An Earliest Due Date (EDD) schedule is a schedule $\pi: J \to \{1,\ldots,n\}$ such that all early jobs are scheduled before all tardy jobs, and all early jobs are scheduled in non-decreasing order of due dates.
\begin{lemma}[\cite{LawlerMoore}]
\label{lem:EDD}%
Any instance of \Prob has an optimal schedule which is EDD.
\end{lemma}

The $d_\#$-many due dates in our instance partition the input set of jobs $J$ in a natural manner: Let $d^{(1)} < \cdots < d^{(d_{\#})}$ denote the $d_{\#}$ different due dates of the instances, and define $J_i=\{j : d_j = d^{(i)}\}$ for each $i \in \{1,\ldots,d_{\#}\}$. The first step of our algorithm is to treat each of the $J_i$'s as a separate \SDProb instance, and use the Cygan \emph{et al.}~\cite{CyganMWW19} algorithm stated in Theorem~\ref{thm:CyganUpperBound} to solve each of these instances. Consider the instance corresponding to some $J_i$. The Cygan \emph{et al.} algorithm has a useful property that it produces an integer vector $A = (A[k])^{d^{(i)}}_{k=0}$ where the $k$th entry $A[k]$ is equal to the maximum weight of any subset of early jobs in $J_i$ with total processing time at most $k$. 
We call such an integer vector a \emph{solution vector} for the \SDProb instance $J_i$. We can extend this notion naturally to a solution vector for a general \Prob instance. 
\begin{definition}
A solution vector for a \Prob instance $J$ is an integer vector $A = (A[k])^{d_{\max}}_{k=0}$ where the $k$'th entry $A[k]$ is equal to the maximum weight of any subset of early jobs in $J$ with total processing time at most $k$.
\end{definition}

Algorithm~\ref{alg:Main} provides a full description of our algorithm. It successively combines the solution vectors $A_1,\ldots,A_{d_{\#}}$ corresponding to $J_1,\ldots,J_{d_{\#}}$ into a single solution vector $A$ for $J$ using $(\max,+)$-convolutions. The last entry $A[d_{\max}]$ of $A$ will contain the maximum weight of early jobs in the instance, so the algorithm returns $\sum_j w_j -A[d_{\max}]$. 
 
\begin{algorithm}
\caption{}%\textsc{Algorithm FastScheduler}} 
\label{alg:Main}
\begin{algorithmic}[1]
\Input{A \Prob instance.}
\Output{The minimum weighted number of tardy jobs.}
\State Compute $J_1,\ldots,J_{d_\#}$.
\State Compute solution vectors $A_1,\ldots,A_{d_{\#}}$ corresponding to $J_1,\ldots,J_{d_\#}$. % using~\cite{BateniHSS18}.
\State $A=A_1$.
\State \textbf{for} {$i=2,\ldots,d_\#$} \textbf{do} $A=A \oplus A_i$.
\State \textbf{return} $\sum_j w_j-A[d_{\max}]$.
\end{algorithmic}
\end{algorithm}

\begin{lemma}
\label{lem:Correctness}%
Algorithm~\ref{alg:Main} correctly returns the minimum total weight of tardy jobs.
\end{lemma}

\begin{proof}
We prove the lemma by induction on the $d_{\#}$ iterations preformed in line 4 of the algorithm (where the first iteration is line 3). For an iteration $i \in \{1,\ldots,d_{\#}\}$, we claim that any entry $A[k]$, $k \in \{0,\ldots,d^{(i)}\}$, at iteration $i$ contains the maximum total weight of early jobs with total processing time at most $k$ in an EDD schedule for $J_1 \cup \cdots \cup J_i$. According to Lemma~\ref{lem:EDD}, such a schedule is optimal for $J_1 \cup \cdots \cup J_i$ (under the constraint that the total processing time is at most $k$), and so proving this claim shows that $w-A[d_{\max}]$ is indeed the minimum total weight of tardy jobs at the end of the algorithm.

For $i=1$ the claim is true due to the correctness of the specific \SDProb algorithm chosen (\emph{e.g.}, the Cygan \emph{et al.}~\cite{CyganMWW19} algorithm). Assume therefore that $i>1$, and choose some arbitrary $k \in \{1,\ldots,d^{(i)}\}$. Let $E_1,\ldots,E_i$ denote the set of early jobs in an optimal EDD schedule for $J_1 \cup \cdots \cup J_i$ with total processing time at most $k$, and let $a$ denote the total weight of these jobs. Furthermore, let $k'$ and $a'$ denote the total processing time and weighted number of the jobs in $E_1,\ldots,E_{i-1}$ respectively. Then by induction we have $A'[k'] \geq a'$, where~$A'$ is the solution vector for iteration $i-1$, as $E_1,\ldots,E_{i-1}$ is a set of early jobs in $J_1\cup \cdots \cup J_{i-1}$ with total processing time $k'$. Similarly, by correctness of the algorithm in~\cite{BateniHSS18}, we know that $A_i[k-k'] \geq a-a'$. Hence, $A[k] \geq A'[k'] + A_i[k-k']=a' + (a-a') = a$. 

Conversely, suppose that $A[k]=A'[k'']+A_i[k-k'']$ for some $k'' \in \{0,\ldots,k\}$. Then there is a set of early jobs in $J_1\cup \cdots \cup J_{i-1}$ with $A'[k'']$ weighted number of tardy jobs and total processing time $k'' \leq d^{(i-1)}$, and a set of early jobs in $J_i$ with $A_i[k-k'']$ weighted number of tardy jobs and total processing time $k-k'' \leq d^{(i)}-k''$. Combining these two jobs together yields a set of early jobs in $J_1 \cup \cdots \cup J_i$ with $A[k]$ weighted number of early jobs. Hence, by optimality of $E_1,\ldots,E_i$, we have $A[k] \leq a$. All together this shows that $A[k]=a$, and so the lemma holds.
\end{proof}

\begin{lemma}
\label{lem:RunningTime}%
Algorithm~\ref{alg:Main} runs in $\widetilde{O}(n+d_{\#}d_{\max}^2)$ time.
\end{lemma}

\begin{proof}
The first line of the algorithm requires $O(n)$ time. Computing the solution vector $A_i$ in line 2 requires $\widetilde{O}(n_i+(d^{(i)})^2)=\widetilde{O}(n_i+d_{\max}^2)$ time, $n_i=|J_i|$, by Theorem~\ref{thm:CyganUpperBound}. Finally, computing the convolution $A \oplus A_i$ in each iteration of line 4 can be done straightforwardly in $O(d^{(i-1)} \cdot d^{(i)})= O(d_{\max}^2)$ time. Altogether, this gives us a bound of $\sum_i \widetilde{O}(n_i+d_{\max}^2)=\widetilde{O}(n+d_{\#}d_{\max}^2)$ for the entire algorithm. 
\end{proof}

Combining Lemma~\ref{lem:Correctness} and Lemma~\ref{lem:RunningTime} completes the proof of Theorem~\ref{thm:MainAlgorithm}. In the next sections we show how to improve the running time bound given in Lemma~\ref{lem:RunningTime} using different techniques to compute $(\max,+)$-convolutions.

\section{Speedup via the prediction technique}
\label{sec:prediction} 

In this section we show how to improve the running time of our algorithm, under certain parameter ranges, by using the prediction technique of Bateni \emph{et al.}~\cite{BateniHSS18}. In particular, we provide a proof for Theorem~\ref{thm:MainSpeedup}. The key ingredient is that we can speed the $(\max,+)$-convolutions of the algorithm by extending the ideas of Bateni \emph{et al.} to the multiple due date setting. 
\begin{lemma}
\label{lem:PredictionMain}%
One can compute the $(\max,+)$-convolution at any iteration of Algorithm~\ref{alg:Main}  in $\widetilde{O}(n+d_{\#}d_{\max}w_{\max})$ time.
\end{lemma}

\begin{proof}[Proof of Theorem~\ref{thm:MainSpeedup}]
Lemma~\ref{lem:PredictionMain} proves that Algorithm~\ref{alg:Main} correctly computes the minimum weighted number of tardy jobs in any given \Prob instance. To bound its running time, note that lines~1 and~2 of the algorithm can be performed in $\widetilde{O}(n+d_{\#}d_{\max}w_{\max})$ time by using the Knapsack algorithm of Bateni \emph{et al.}~\cite{BateniHSS18}. Moreover, each iteration in line 4 requires $\widetilde{O}(n+d_{\#}d_{\max}w_{\max})$ time according to Lemma~\ref{lem:PredictionMain}. This gives us a total of $\widetilde{O}(d_{\#}n + d^2_{\#}d_{\max}w_{\max})$
for all iterations all together. As this running time dominates the time required to perform lines~1 and~2, this gives us the overall running time of our algorithm. 
\end{proof}

The remainder of the section is devoted to proving Lemma~\ref{lem:PredictionMain}. The key tool used in the prediction technique by Bateni \emph{et al.} is encapsulated in the following theorem: 
\begin{theorem}[\cite{BateniHSS18}]
\label{thm:Bateni}
Let $A = (A[k])^m_{k=0}$ and $B = (B[\ell])^n_{\ell=0}$ be two integer vectors with $m \leq n$, and let~$e$ be some positive integer. Assume we have $m$ intervals $[x_k,y_k] \subseteq [0,n]$, $1 \leq k \leq m$, such that 
\begin{itemize}
\item $A[k]+B[\ell] \geq (A \oplus B)[k+\ell] - e$ for all $k \in \{1,\ldots,m\}$ and $\ell \in [x_k,y_k]$,
\item for all $\ell \in \{1,\ldots,n\}$ there exists an $k \in \{1,\ldots,m\}$ such that $A[k]+B[\ell-k]=(A \oplus B)[\ell]$ and $\ell-k \in [x_k,y_k]$.
\item $x_k \leq x_{k+1}$ and $y_k \leq y_{k+1}$ for all $0 \leq k < m$.
\end{itemize}
Then $A \oplus B$ can be computed in $\widetilde{O}(ne)$ time.
\end{theorem}

For two integer vectors $A$ and $B$, we refer to a set of $|A|$ intervals that satisfy all requirements of Theorem~\ref{thm:Bateni} with parameter~$e$ as \emph{range intervals of $A$ in $B$ with error $e$}. By Theorem~\ref{thm:Bateni} above, to prove Lemma~\ref{lem:PredictionMain} it is enough to prove the following:
%Let $i^\star \in \{2,\ldots,d_{\#}\}$ denote some iteration in Algorithm~\ref{alg:Main}. Let $A$ be a solution vector for the \Prob instance formed by all jobs in $J_A=J_1 \cup \ldots \cup J_{i^\star-1} = \{j \in J: d_j \le d^{(i^\star-1)}\}$, and let~$B$ denote the solution vector for the \SDProb instance $J_B=J_{i^\star}=\{j \in J : d_j = d^{(i^\star)}\}$. By Theorem~\ref{thm:Bateni} above, to prove Lemma~\ref{lem:PredictionMain} it is enough to prove the following:
%\begin{lemma}
%\label{lem:rangeinterval}
%There is an algorithm running in $\widetilde{O}(i^\star d^{(i^\star)}+n)$ time that computes range intervals of $A$ in $B$ with error $e\in O(i^\star\cdot w_{\max})$.
%\end{lemma}
\begin{lemma}
\label{lem:rangeinterval}
Consider some iteration $i \in \{1,\ldots,d_{\#}\}$ in Algorithm~\ref{alg:Main}, and let $A$ and $B=A_i$ denote the two vectors to be convoluted at this iteration. There is an algorithm running in $\widetilde{O}(n+d_{\#} d_{\max})$ time that computes range intervals of $A$ in $B$ with error $e = O(d_{\#}w_{\max})$.
\end{lemma}
We provide a complete proof of Lemma~\ref{lem:rangeinterval} below.

\subsection{Fractional Solution Vectors}
\label{sec:rangewmax}%

Roughly speaking, the range intervals for $A$ in $B$ tell us which entries $A[k]$ and $B[\ell]$ sum up to a value close to $C[k+\ell]$. To be able to determine this without computing $C$, we need a good approximation for $C$. Similar to Bateni \emph{et al.}~\cite{BateniHSS18}, we make use of the solution to the \emph{fractional} version of our problem to obtain an approximation for the (non-fractional) solution.
\begin{definition}
In the fractional \Prob problem we are given a set of jobs $J=\{1,\ldots,n\}$, with processing times $(p_j)_{j=1}^n$, weights $(w_j)_{j=1}^n$, and due dates $(d_j)_{j=1}^n$, and our goal is to compute~$n$ real values $0 \leq x_1, \ldots, x_n \leq 1$ such that 
\begin{itemize}
\item $\sum_{d_k \leq d_j } p_kx_k\le d_j$ for all $j \in J$, and
\item $\sum_{j \in J} w_jx_j$ is maximized. 
\end{itemize}
\end{definition}
%\noindent Note that this definition uses the assumption of Lemma~\ref{lem:EDD} that there always exists an optimal EDD schedule. 
\noindent Similar to the 0/1 version of \Prob, we define a \emph{fractional solution vector} for a fractional \Prob instance $J$ as a vector $A'$ of size $d_{\max} = \max\{d_j : j \in J\}$, where $A'[k]$ equals the maximum value $\sum w_jx_j$ over all feasible solutions that satisfy the additional constraint $\sum p_jx_j\le k$, for each $k \in \{0,\ldots,d_{\max}\}$. 

Algorithm~\ref{alg:fracsolvec} below computes a solution vector for a given fractional \Prob instance. It is a modified version of the standard greedy algorithm for computing fractional solutions for Knapsack~\cite{dantzig1957discrete} used by Batani \emph{et al.}~\cite{BateniHSS18}, that is able to deal with the presence of more than one due date.
The algorithm assumes the jobs are sorted according to the \emph{Weighted Shortest Processing Time (WSPT)} rule, \emph{i.e} in non-increasing values of $w_j/p_j$. This is a standard technique in many scheduling algorithms, and can be performed by any $O(n \log n)$ sorting algorithm. Moreover, in the description of the algorithm, we assume that $w_j/p_j = 0$ whenever $j > n$.

In the algorithm we keep track of the total processing times $p^{(i)}$ of (fractional) jobs that are scheduled early from $J_1 \cup \cdots \cup J_i$. The algorithm then iterates over the WSPT sorted list of jobs, and repeatedly adds a slice of unit processing from the current job to the next entry in the solution vector. Before adding the slice, the algorithm checks that no due dates will be violated by adding the slice using the current values of the $p^{(i)}$'s. If adding the current slice violates some due date, the algorithm skips to the next job in the list. 

\begin{algorithm}
\caption{}
\label{alg:fracsolvec}
\begin{algorithmic}[1]
\Input{WSPT sorted set of jobs $J=\{1,\ldots,n\}$ with $d_{\char"0023}$ different due dates.}
\Output{A fractional solution vector $A'$ for $J$.}
\State Let $\{d^{(1)},\ldots,d^{(d_{\#})}\}=\{d_j : j \in J\}$ be the different due dates in $J$. 
\State Set $p^{(i)}=0$ for all $1\le i\le d_{\#}$.
\State Set $A'[0]=0$, $j=1$, and $p^\star_j=0$.
\For{$k=1,\ldots,d^{(d_{\#})}$}\label{line:firstfor}
\State Let $i \in \{1,\ldots,d_{\#}\}$ be such that $d_j=d^{(i)}$.
\If{$\min\{p_j-p^\star_j,d^{(i)}-p^{(i)}, d^{(i+1)}-p^{(i+1)}, \ldots, d^{(d_{\#})}-p^{(d_{\#})}\} > 0$}  \label{line:check}
\State $A'[k]=A'[k-1]+w_j/p_j$.
\State $p^\star_j=p^\star_j+1$.\label{line:pincrease}
\State \textbf{for} $i' = i,\ldots,d_{\#}$ \textbf{do} 
$p^{(i')}=p^{(i')}+1$.\label{line:secondfor}
\EndIf
\State \textbf{otherwise} $j=j+1$ and $p^\star_j=0$.
\EndFor
\State \textbf{return} $A'$.
\end{algorithmic}
\end{algorithm}

%\begin{lemma}
%\label{lem:fracsolvec}
%Given a fractional \Prob instance with $n$ jobs, $i^\star$ many different due dates, and maximum due date $d^{(i^\star)}$, sorted according to the WSPT rule, Algorithm~\ref{alg:fracsolvec} correctly computes a solution vector for this instance in $\widetilde{O}(i^{\star}d^{(i^\star)} + n)$ time. 
%\end{lemma}

\begin{lemma}
\label{lem:fracsolvec}
Given a fractional \Prob instance with $n$ jobs, $d_{\#}$ many different due dates, and maximum due date $d_{\max}$, sorted according to the WSPT rule, Algorithm~\ref{alg:fracsolvec} correctly computes a solution vector for this instance in $\widetilde{O}(n+d_{\#}d_{\max})$ time. 
\end{lemma}

\begin{proof}
It is straightforward to verify that Algorithm~\ref{alg:fracsolvec} runs in $\widetilde{O}(n+d_{\#}d_{\max})$ time. To prove its correctness, we argue that $A'$ is a correct solution vector for the fractional \Prob instance~$J$. For a given fractional solution $0 \leq x_1,\ldots,x_n \leq 1$, we define the \emph{length} of this solution to be the value $\sum_j p_jx_j$. An \emph{optimal fractional solution of length $k$} is a feasible fractional solution with maximum value $\sum w_jx_j$ among all solutions of length $k$. 

Consider some iteration $k \in \{0,\ldots,d_{\max}\}$ of the algorithm, and let $j$ be the current job under consideration. We argue by induction on $k$ that there exists an optimal fractional solution $x_1,\ldots,x_n$ of length $k$, with value $A'[k]$, such that at iteration~$k$ we have 
\begin{equation}
\label{eqn:pi}%
p^{(i)} \quad=\sum_{\ell \leq j, \,d_\ell \leq d^{(i)}} p_\ell x_\ell \quad \leq \quad d^{(i)}   
\end{equation}
for each $i \in \{1,\ldots,d_{\#}\}$. For $k = 0$ this is clearly the case for the solution $x_1 = \cdots = x_n = 0$, as $A'[0]=0$. So assume $k > 0$, and let $i \in \{1,\ldots,d_{\#}\}$ be such that $d_{j}=d^{(i)}$. Moreover, let $x_1,\ldots,x_n$ be the fractional solution of length at most $k -1$ which is guaranteed by induction. Observe that line~\ref{line:check} ensures that we can add extra unit of processing time unit to all $p^{(i')}$'s with $i' \geq i$ without violating~(\ref{eqn:pi}). It also ensures that $x_j+1/p_j \leq 1$, since $p_j-p^\star_j=p_j - p_j\cdot x_j >0$. Thus, by setting $x_j=x_j+1/p_j$, we obtain a fractional solution that satisfies (\ref{eqn:pi}) after the $p^{(i)}$'s are updated in line~\ref{line:secondfor} at the end of the iteration. It follows that our new fractional solution remains feasible, and it clearly has length at most $k$. Therefore, what remains to show is that this solution is optimal among all feasible solutions of length $k$.

First observe that by the description of the algorithm, increasing any value of some $x_\ell$ with $\ell < j$ makes the solution $x_1,\ldots,x_n$ infeasible. Thus, to improve our solution, we can only increase the value of some $x_\ell$ with $\ell \geq j$. Since the jobs are ordered according to the WSPT rule, the current job $j$ has the best weight per processing time ratio among all jobs in $\{j,\ldots,n\}$. Thus, scheduling a processing time unit of a combination of different jobs in $\{j+1,\ldots,n\}$ cannot increase the total value of the solution by more than $w_j/p_j$. Moreover, increasing $x_j$ by more than $1/p_j$ results in a solution of length greater than $k$, and so $x_\ell>0$ for some other job $\ell < j$ must be decreased. However, this cannot increase $\sum_{j \in J} w_jx_j$ either, because all jobs $\ell < j$  have at least the same weight per processing time as job $j$. It follows that $x_1,\ldots,x_n$ has the maximum value $\sum w_jx_j$ among all feasible solutions of length~$k$.
\end{proof}

Each entry of the fractional solution vector clearly upper bounds the (integral) solution vector component-wise. That is, if $A$ is a solution vector of some \Prob instance, and $A'$ is the corresponding fractional solution vector which is computed by Algorithm~\ref{alg:fracsolvec} on this instance, then $A'[k] \geq A[k]$ for each entry $k$. In the next lemma we show that each $A'[k]$ is not too far away from~$A[k]$.%, by showing that we can round down a fractional solution while not loosing too much from 

%The total weight of all early jobs in an optimal fractional solution clearly upper-bounds the value of all early jobs in an optimal integral solution. We show that from the fractional solution, we can easily (and quickly) compute a feasible integral solution. The value of all early jobs in a feasible integral solution clearly lower-bounds the value of all early jobs in an optimal integral solution. Hence, to show that the optimal fractional solution is a good approximation for the optimal integral solution, we can show that the optimal fractional solution is a good approximation for the feasible integral solution that we compute from it.

%As outlined before, the main reason for computing a solution vector for a fractional \Prob instance is that it provides a good approximation of the solution vector for the integral (\emph{i.e.} non-fractional) instance.  This is proven in Lemma~\ref{lem:goodapprox} below.
%Note that we can easily obtain a feasible solution by removing the fractional jobs from the schedule corresponding to the fractional solution vector. The main idea in the proof of Lemma~\ref{lem:goodapprox} is to show that by removing the fractional jobs, we decrease the total weight of early jobs by at most $w_{\max}$ per removed job.

\begin{lemma}
\label{lem:goodapprox}%
Let $J$ be a given \Prob instance with $d_{\#}$ many different due dates, and maximum weight $w_{\max}$. Furthermore, let $A$ a solution vector for this instance, and let $A'$ be the fractional solution vector for $J$. Then 
$$
0 \leq A'[k]-A[k]  \leq d_{\#} w_{\max}
$$
for every entry $k$ in $A$ and $A'$. 
\end{lemma}
\begin{proof}
We claim there exists an optimal fractional solution $(x_{j})_{{j}=1}^n$ for $J$ such that at most~$d_{\#}$ values of the fractional solution are non-integer. More specifically, we show that the optimal fractional solution $(x_{j})_{{j}=1}^n$ implicitly computed by Algorithm~\ref{alg:fracsolvec} has this property. Note that whenever $\min\{p_j-p^\star_j,d^{(i)}-p^{(i)}, d^{(i+1)}-p^{(i+1)}, \ldots, d^{(d_{\#})}-p^{(d_{\#})}\} = 0$ for some job $j$ with deadline $d_j=d^{(i)}$, we move to the next job. If $p_j-p^\star_j=0$, then we implicitly have $x_j=1$ since all of job $j$ is scheduled. If this is not the case, we have that $d^{(i')}-p^{(i')}=0$ for some $i'\ge i$. Assume that $j$ is the first job with deadline $d^{(i)}$ where this happens. It follows that no unit of processing time of any other job with deadline $d^{(i)}$ is scheduled. Consequently we have that for every deadline $d^{(i)}$, at most one non-integer fraction of a job with that deadline is in the optimal fractional solution. Hence, the total number of fractional jobs in the optimal fractional solution is at most $d_{\#}$. Removing these jobs decreases the total weight of early jobs by at most $d_{\#}\cdot w_{\max}$ and yields an integer solution for~$J$. It follows that $A'[k]-A[k]\le d_{\#} w_{\max}$ for all $k$.
\end{proof}

\subsection{Constructing the range intervals}

Consider some iteration $i \in \{2,\ldots,d_{\#}\}$ of Algorithm~\ref{alg:Main}. Let $A$ be a solution vector for the \Prob instance formed by all jobs in $J_A=J_1 \cup \ldots \cup J_{i-1}$ computed at iteration $i-1$, and let~$B$ denote the solution vector for the \SDProb instance $J_B=J_i$. Having Lemma~\ref{lem:goodapprox} in place, we can describe how to construct the range intervals of $A$ in $B$. We begin by first using Algorithm~\ref{alg:fracsolvec} to compute the fractional solution vectors $A'$, $B'$, and $C'$ corresponding to the instances $J_A$, $J_B$, and $J_C=J_A \cup J_B$, respectively. Next we define a set of intervals $[x_k,y_k]_{k=0}^{|A|}$ as follows.
\begin{align*}
x_k&=\min\{\ell: C'[k+\ell]-(A'[k]+B'[\ell])\le 2i \cdot w_{\max}\}, \text{ and}\\
y_k&=\max\{\ell: C'[k+\ell]-(A'[k]+B'[\ell])\le 2i \cdot w_{\max}\}.
\end{align*}
Using analogous arguments as is done by Bateni et al.~\cite{BateniHSS18}, we prove the following:%can show that the set of $|A|$ intervals defined above fulfill the requirements from Theorem~\ref{thm:Bateni} with an error $e = O(i\cdot w_{\max})$. For completeness, we also provide a proof here.
\begin{lemma}
\label{lem:rangeintervals}%
The set of intervals $\{[x_k,y_k] : 0 \leq k \leq |A|\}$ defined above are range intervals of $A$ in~$B$ with an error of $e = O(i\cdot w_{\max})$.
\end{lemma}

\begin{proof}
We begin by first showing that set of intervals $\{[x_k,y_k] : 0 \leq k \leq |A|\}$ satisfy the first two conditions of Theorem~\ref{thm:Bateni}. Consider some $k \in \{0,\ldots,|A|\}$, and let $\Delta_k(\ell)=C'[k+\ell]-(A'[k]+B'[\ell])$. 
By Lemma~\ref{lem:goodapprox}, we have that 
\[
\Delta_k(\ell)\le 2i\cdot w_{\max} \Longrightarrow C[k+\ell]-(A[k]+B[\ell]) \le 4i\cdot w_{\max}.
\]

Define $\delta(k)$ as the smallest integer such that $C'[k+\delta(k)]=A'[k]+B'[\delta(k)]$. 
Now consider $1\le \ell\le \delta(k)$. We argue that $\Delta_k(\ell-1)\ge\Delta_k(\ell)$. Using the definition of $\Delta_k$ we can rewrite this as $C'[k+\ell]-C'[k+(\ell-1)]\leq B'[\ell]-B'[\ell-1]$. Recall that due to the description of Algorithm~\ref{alg:fracsolvec}, we have $C'[k+\ell]-C'[k+(\ell-1)]=w_{j_c}/p_{j_c}$ for some $j_c \in J_A \cup J_B$, and $B'[\ell]-B'[\ell-1]=w_{j_b}/p_{j_b}$ for some $j_b \in J_B$. Now, observe that for $\delta(k)$, the total processing time of jobs from $J_A$ (resp. $J_B$) scheduled in the fractional solution for $C'[k+\ell]$ is exactly $k$ (resp.~$\delta(k)$). Thus, since Algorithm~\ref{alg:fracsolvec} never removes jobs in its computation, and since $\ell \leq \delta(k)$, we know that the total processing time of jobs from~$J_A$ scheduled in the fractional solution for $C'[k+(\ell-1)]$ is at most $k$, which means that the total processing time of jobs from $J_B$ in this solution is at least $\ell-1$. Thus, if $j_c \in J_B$ we have $w_{j_c}/p_{j_c} \leq w_{j_b}/p_{j_b}$. This is because in the fractional solution corresponding to $B'[\ell-1]$, the total processing time of jobs from $J_B$ is exactly $\ell-1$, and the jobs are sorted according to the WSPT rule. If, on the other hand, $j_c \in J_A$, then the total processing time of jobs from~$J_B$ scheduled in the fractional solution for $C'[k+(\ell-1)]$ is at least $\ell$, which means that job $j_b$ is already scheduled in the fractional solution for $C'[k+(\ell-1)]$. So again we get $w_{j_c}/p_{j_c} \leq w_{j_b}/p_{j_b}$, implying that $C'[k+\ell]-C'[k+(\ell-1)]\leq B'[\ell]-B'[\ell-1]$.

Thus, we have that $\Delta_k(\ell-1)\ge\Delta_k(\ell)$ for all $1\le \ell\le \delta(k)$.
By an analogous argument we can show that $\Delta_k(\ell-1)\le\Delta_k(\ell)$ holds for all $\delta(k)\le \ell\le |B'|$. It follows that $\Delta_k$ is non-increasing in the interval $[0,\delta(k)]$, and non-decreasing in the interval $[\delta(k), |B'|]$. From this, we immediately get that the first two requirements from Theorem~\ref{thm:Bateni} are fulfilled for the interval $[x_k,y_k]$ with an error parameter of~$e=4i\cdot w_{\max}$.

To show that the third requirement from Theorem~\ref{thm:Bateni} is fulfilled, we argue that $\Delta_{k+1}(x_k)\ge \Delta_{k}(x_k)$. Using the definition of $\Delta_k$ we can rewrite this as $C'[k+1+x_k]-C'[k+x_k]\ge A'[k+1]-A'[k]$. Note that $x_k\le\delta(k)$. Hence, we know that $C'[k+1+x_k]$ and $C'[k+x_k]$ contain jobs from $J_A$ of total processing time at most $k+1$ and $k$, respectively. Let job $j$ be added by Algorithm~\ref{alg:fracsolvec} when computing $C'[k+1+x_k]$ from $C'[k+x_k]$, and let job $j'$ be added Algorithm~\ref{alg:fracsolvec} when computing $A'[k+1]$ from $A'[k]$. It follows that $w_{j}/p_{j}\ge w_{j'}/p_{j'}$. By an analogous argument we get that $\Delta_{k+1}(y_k)\ge \Delta_{k}(y_k)$. It follows that the third requirement from Theorem~\ref{thm:Bateni} is fulfilled.
\end{proof}

%We can compute range intervals of $A$ in $B$ in $O(|B|)=O(d^{(i^\star)})$ time using Algorithm~\ref{alg:ranges}. To do so, we exploit the third requirement from Theorem~\ref{thm:Bateni}. This improves the computation time of the range intervals by a factor of $\log(d^{(i^\star)})$ when compared to the approach by Bateni et al.~\cite{BateniHSS18}.

\begin{algorithm}
\caption{}%\textsc{Algorithm FastScheduler}} 
\label{alg:ranges}
\begin{algorithmic}[1]
\Input{Vectors $A',B',C'$, number of deadlines $i^\star$, and the maximum weight $w_{\max}$.}
\Output{Range intervals $\{[x_k,y_k] : 0 \leq k \leq |A| \}$ of $A$ in $B$.}
\State Set $x=y=0$.
\For{$k=1$ to $|A'|$}
\State \textbf{while} $C'[k+x]-A'[k]-B'[x]>2i\cdot w_{\max}$ \textbf{do} $x= x+1$.
\State \textbf{if} $k=1$ \textbf{then} $y=x$.
\State \textbf{while} $C'[k+y]-A'[k]-B'[y]\le2i\cdot w_{\max}$ \textbf{do} $y= y+1$.
\State Set $[x_k,y_k]=[x,y-1]$.
\EndFor
\end{algorithmic}
\end{algorithm}

\begin{lemma}
\label{lem:ranges}%
Algorithm~\ref{alg:ranges} correctly computes the range intervals $[x_i,y_i]_{i=0}^{|A|}$ of $A$ in $B$ in $\widetilde{O}(d_{\max})$ time.
\end{lemma}
\begin{proof}
It is straightforward to verify that Algorithm~\ref{alg:ranges} runs in $O(|A|+|B|)=O(d^{(i)})=\widetilde{O}(d_{\max})$ time. For its correctness, recall the function $\Delta_k(\ell)=C'[k+\ell]-A'[k]-B'[\ell]$ used in the proof of Lemma~\ref{lem:rangeintervals}. As is shown in this proof, for every $k$ there exists an $\delta(k)$ such that $\Delta_k(\ell)$ is non-increasing in the interval $[0,\delta(k)]$, and $\Delta_k(\ell)$ is non-decreasing in the interval $[\delta(k),|B'|]$. Hence, to find $[x_1,y_1]$, Algorithm~\ref{alg:ranges} first finds~$x_1$ by iterating over values for $\ell$ starting from $0$, thereby finding the smallest integer $\ell$ such that $\Delta_1(\ell)\le 2i\cdot w_{\max}$. Afterwards, Algorithm~\ref{alg:ranges} finds $y_1$ by iterating over values for $\ell$ and finding the largest integer~$\ell$ such that $\Delta_1(\ell)\le 2i\cdot w_{\max}$. We know that in the interval $[x_1,\delta(1)]$ the function $\Delta_1(\ell)$ is non-increasing, hence the largest integer $\ell$ such that $\Delta_1(\ell)\le 2i\cdot w_{\max}$ must lie in the interval $[\delta(1),|B'|]$ and, more specifically, once $\Delta_1(\ell')>2i\cdot w_{\max}$ for some $\ell' \in [\delta(1),|B'|]$, we know that $\Delta_1(\ell'')>2i\cdot w_{\max}$ for all $\ell''>\ell'$. Hence, Algorithm~\ref{alg:ranges} can stop the search for $y_1$ after encountering the first value $\ell$ such that $\Delta_1(\ell)>2i\cdot w_{\max}$. Furthermore, we know that for all $k$ we have that $x_k\le x_{k+1}$ and $y_{k}\le y_{k+1}$,  since the range intervals fulfill the third requirement of Theorem~\ref{thm:Bateni}. It follows that for any $k > 1$, Algorithm~\ref{alg:ranges} can start the search for $x_k$ at $x_{k-1}$ and the search for $y_i$ at $y_{k-1}$. Thus, Algorithm~\ref{alg:ranges} indeed correctly computes the range intervals of $A$ in~$B$, and the lemma follows.
\end{proof}

\begin{proof}[Proof of Lemma~\ref{lem:rangeinterval}]
Computing all three fractional solution vectors $A'$, $B'$, and $C'$ can be done in $\widetilde{O}(n+d_{\#}d_{\max})$ time according to Lemma~\ref{lem:fracsolvec}. From these we can compute the range intervals of $A$ in $B$ in $\widetilde{O}(d_{\max})$ time due to Lemma~\ref{lem:ranges}. These intervals satisfy the conditions of Theorem~\ref{thm:Bateni} with an error parameter $e = O(i \cdot w_{\max})=O(d_{\#} w_{\max})$ as is proven in Lemma~\ref{lem:rangeintervals} (which in turn relies on Lemma~\ref{lem:goodapprox}). All together this gives us an algorithm for computing the range intervals of $A$ in $B$ in $\widetilde{O}(n+d_{\#}d_{\max})$ time.
\end{proof}

\section{Further Speedups}

In this section we show how to obtain further speedups to Algorithm~\ref{alg:Main}, by using other techniques for fast $(\max,+)$-convolutions. In particular, we provide a complete proof for Theorem~\ref{thm:SecondSpeedup} via Lemma~\ref{lem:Speedup1}, \ref{lem:Speedup2}, and~\ref{lem:Speedup3} proven below.

\subsection{Using \boldmath{$s$}-step concave solutions vectors}

The main result of this subsection is given in the following lemma:
\begin{lemma}
\label{lem:ConcaveMain}%
Let $i \in \{2,\ldots,d_{\#}\}$. One can compute the $(\max,+)$-convolution in iteration $i$ of Algorithm~\ref{alg:Main} in $\widetilde{O}(n_i+d_{\max}p_{\max})$ time, where $n_i = |J_i|$.
\end{lemma}

\begin{lemma}
\label{lem:Speedup1}%
Algorithm~\ref{alg:Main} can be implemented to run in $\widetilde{O}(n+d_{\#}d_{\max}p_{\max})$ time. 
\end{lemma}
\begin{proof}
We can compute the solution vector of each \SDProb instance $J_i$, $1 \leq i \leq d_{\#}$, using the $\widetilde{O}(n_i+d_{\max}p_{\max})$ algorithm of Axiotis and Tzamos~\cite{AxiotisTzamos19}. By using Lemma~\ref{lem:ConcaveMain}, we can perform each iteration $i$ of the algorithm in $\widetilde{O}(n_i+d_{\max}p_{\max})$ time. Thus, the running time of the entire algorithm can be bounded by $\widetilde{O}(\sum_i (n_i+d_{\max}p_{\max}))=\widetilde{O}(n+d_{\#}d_{\max}p_{\max})$.
\end{proof}

We proceed to describe how to perform the $i$th iteration of Algorithm~\ref{alg:Main} in $\widetilde{O}(n_i+d_{\max}p_{\max})$ time (Lemma~\ref{lem:ConcaveMain}), using the techniques of Axiotis and Tzamos.
\begin{definition}
An integer vector $B = (B[\ell])^n_{\ell=0}$ is \emph{concave} if for all $1 \leq \ell \leq n$ we have $B[\ell] - B[\ell-1] \geq B[\ell+1] - B[\ell]$. The vector $B$ is called \emph{$s$-step concave} if the vector $B[0]B[s]B[2s] \cdots$ is concave, and for all $\ell$ such that $\ell \mod s \neq 0$ we have $B[\ell]=B[\ell-1]$.
\end{definition}

\begin{theorem}[\cite{AxiotisTzamos19}]
\label{thm:Concave}
Let $A = (A[k])^m_{k=0}$ and $B = (B[\ell])^n_{\ell=0}$ be two integer vectors with $m \leq n$. If $B$ is $s$-step concave for any $s \in \{1,\ldots,n\}$, then $A \oplus B$ can be computed in $O(n)$ time.
\end{theorem}

Let $A$ be the solution vector for the jobs sets $J_1 \cup \cdots \cup J_{i-1}$ computed at iteration $i-1$, and let $B$ denote the solution vector of $J_i$. To compute $A \oplus B$, we first partition $J_i$ into subsets $J_{i,1},\ldots,J_{i,p_{\max}}$, where all jobs in $J_{i,p}$ have processing time $p$ for each $p \in \{1,\ldots,p_{\max}\}$. We next compute the solution vector $B_p$ corresponding to each \SDProb instance $J_{i,p}$. Finally, we compute $A \oplus B = ( \cdots ((A \oplus B_1) \oplus B_2) \cdots ) \oplus B_{p_{\max}})$. 

\begin{algorithm}
\caption{}%\textsc{Algorithm FastScheduler}} 
\label{alg:Concave}
\begin{algorithmic}[1]
\Input{Solution vectors $A$ and $B$ corresponding to $J_1 \cup \cdots \cup J_{i-1}$ and $J_i$, respectively.}
\Output{$C= A \oplus B$.}
\State Compute $J_{i,p} = \{j \in J_i : p_j = p\}$ for each $p \in \{1,\ldots,p_{\max}\}$.
\State Compute solution vectors $B_1,\ldots,B_{p_{\max}}$ corresponding to $J_{i,1},\ldots,J_{i,p_{\max}}$. 
\State \textbf{for} {$p=1,\ldots,p_{\max}$} \textbf{do} $A=A \oplus B_p$.
\State \textbf{return} $A$.
\end{algorithmic}
\end{algorithm}

As observed by Axiotis and Tzamos, computing the solution vector $B_p$ correspond to $J_{i,p}$, $1 \leq p \leq p_{\max}$ can be easily done since all jobs in $J_{i,p}$ have the same processing time. In particular, if $w^1 \geq w^2 \geq \cdots \geq w^{n_{i,p}}$ are the weights of $J_{i,p}$, where $n_{i,p} = |J_{i,p}|$, then we have 
$$
B_p[0]=0, \, B_p[p]=w^1, \, B_p[2p]=w^1 + w^2, \, \cdots
$$
Moreover, $B[\ell] = B[\ell-1]$ for any $\ell$ not divisible by $p$. Therefore $B_p$ is a $p$-step concave vector, and so computing $A \oplus B_p$ can be done in $\widetilde{O}(d^{(i)})=\widetilde{O}(d_{\max})$ time. Thus, in total we can compute $A \oplus B$ in $\widetilde{O}(n_i+d_{\max}p_{\max})$ time, and so Lemma~\ref{lem:ConcaveMain} holds.

\subsection{Using \boldmath{$s$}-step concave inverse solutions vectors}

We next prove the following: 
\begin{lemma}
\label{lem:ConcaveSecond}%
Let $i \in \{2,\ldots,d_{\#}\}$. One can compute the $(\max,+)$-convolution in iteration $i$ of Algorithm~\ref{alg:Main} in $\widetilde{O}(n_iw^2_{\max})$ time, where $n_i = |J_i|$.
\end{lemma}

The main idea here is to replace the solution vectors used by our main algorithm with \emph{inverse solution vectors}: 
\begin{definition}
A solution vector for a \Prob instance $J$ with $n$ jobs of maximum weight $w_{\max}$ is an integer vector $A = (A[k])^{nw_{\max}}_{k=0}$ where the $k$'th entry $A[k]$ is equal to the minimum processing time of any subset of early jobs in $J$ with total processing weight at least~$k$.
\end{definition}

By replacing $(\max,+)$-convolutions with $(\min,+)$-convolutions, which are equivalent by taking the negation of the vectors, we can modify Algorithm~\ref{alg:Main} to work also with inverse solution vectors. Let $A$ be the inverse solution vector for the jobs sets $J_1 \cup \cdots \cup J_{i-1}$ computed at iteration $i-1$, and let $B$ denote the inverse solution vector of $J_i$. To compute the $(\min,+)$-convolution between $A$ an $B$, we again  iteratively compute the convolution between $A$ and each inverse solution vector $B_w$ corresponding to the job set $J_{i,w}= \{j \in J_i : w_j =w\}$. As each $B_w$ is $w$-step concave, each convolution can be done in $\widetilde{O}(n_iw_{\max})$ time by using a $(\min,+)$-variant of Theorem~\ref{thm:Concave}. Altogether, this gives us $\widetilde{O}(n_iw^2_{\max})$ time for the entire iteration, so Lemma~\ref{lem:ConcaveSecond} holds. 

\begin{lemma}
\label{lem:Speedup2}%
Algorithm~\ref{alg:Main} can be implemented to run in $\widetilde{O}(n^2+d_{\max}w^2_{\max})$ time. 
\end{lemma}

\begin{proof}
Consider the variant of Algorithm~\ref{alg:Main} that uses inverse solution vectors. We can compute the inverse solution vector of each \SDProb instance $J_i$ in $\widetilde{O}(n_iw^2_{\max})$ time using the Knapsack algorithm of Axiotis and Tzamos~\cite{AxiotisTzamos19}. Then we can compute the convolution of iteration $i$ in $\widetilde{O}(n_iw^2_{\max})$ time according to Lemma~\ref{lem:ConcaveSecond}. Thus, altogether this gives us an algorithm running in time $\widetilde{O}(\sum_i n_iw^2_{\max}) = \widetilde{O}(nw^2_{\max})$.

To obtain the running time in the lemma statement, we distinguish between two cases: If~$n \geq d_{\max}$ then Lawler and Moore's algorithm requires $O(n^2)$ time. Otherwise, if $n < d_{\max}$, the algorithm discussed above requires $\widetilde{O}(d_{\max}w^2_{\max})$ time. Using one of the two algorithms accordingly to whether $n \geq d_{\max}$ or not, gives us the desired running time. 
\end{proof}

\subsection{Using bounded monotone solutions vectors}

As a final speedup to the $(\max,+)$ convolutions preformed in each iteration of Algorithm~\ref{alg:Main}, we prove the following:
\begin{lemma}
\label{lem:BoundedMonotone}%
One can compute the $(\max,+)$-convolution at any iteration of Algorithm~\ref{alg:Main} in $\widetilde{O}((d_{\max}+nw_{\max})^{1.5})$ time.
\end{lemma}

\begin{lemma}
\label{lem:Speedup3}
Algorithm~\ref{alg:Main} can be implemented to run in $\widetilde{O}(n^2+d_{\#}(d_{\max}w_{\max})^{1.5})$ time. 
\end{lemma}

\begin{proof}
As in the proof of Lemma~\ref{lem:Speedup2}, we can assume that $n \leq d_{\max}$ by adding an additional $O(n^2)$ factor to the running time analysis of our algorithm. Thus, as $n \leq d_{\max}$, each iteration of Algorithm~\ref{alg:Main} can be performed in $\widetilde{O}((d_{\max}w_{\max})^{1.5})$ time according to Lemma~\ref{lem:Speedup3}. Moreover, the solution vector $A_i$ of each \SDProb instance $J_i$ can be compute in $\widetilde{O}(d_{\max}w_{\max})$ time by using the Bateni \emph{et al.} algorithm~\cite{BateniHSS18}. Altogether, this gives us $\widetilde{O}(n^2+d_{\#}(d_{\max}w_{\max})^{1.5})$ time for the entire algorithm. 
\end{proof}

The proof of Lemma~\ref{lem:BoundedMonotone} follows directly from the recent breakthrough result by Chi \emph{et al.}~\cite{DChiDX022} for computing the $(\max,+)$-convolution between bounded monotone vectors.
\begin{definition}
An integer vector $A = (A[k])^n_{k=0}$ is \emph{$b$-bounded monotone} if for $1 \leq k < m$ we have $A[k]\le A[k+1]\le b$.
\end{definition}
\begin{theorem}[\cite{DChiDX022}]
\label{thm:BoundedMonotone}
The $(\max,+)$-convolution between two $b$-bounded monotone vectors of length at most $n$ can be computed in $\widetilde{O}((n+b)^{1.5})$ time.
\end{theorem}

Consider some iteration $i \in \{2,\ldots,d_{\#}\}$ of Algorithm~\ref{alg:Main}, and let $A$ be the solution vector for the jobs sets $J_1 \cup \cdots \cup J_{i-1}$ computed at iteration $i-1$, and $B$ be the solution vector of $J_i$. It is easy to see that the values of both $A$ and $B$ are monotonically non-decreasing. Moreover the value of each entry cannot exceed $nw_{\max}$, as this is an upper bound on the total weight of all jobs. Thus, both $A$ and $B$ are $nw_{\max}$-bounded monotone vectors, and so Lemma~\ref{lem:BoundedMonotone} follows directly from Theorem~\ref{thm:BoundedMonotone}.

\section{Lower Bounds}
\label{sec:LowerBound}%

In this section we present our lower bounds for \Prob, namely we prove Theorem~\ref{thm:LowerBound}. We begin by stating the $\forall \exists$ Strong Exponential Time Hypothesis ($\forall \exists$-SETH) on which our lower bounds are based upon. 
\begin{hypothesis}[\cite{AbboudBHS20}]
\label{hyp:FESETH}%
There is no $0 < \alpha  < 1$ and $\varepsilon > 0$ such that for all $k \geq 3$ there is an $O(2^{(1 - \varepsilon)n})$ time algorithm for the following problem: Given a $k$-CNF formula $\phi$ on $n$ variables $x_1,\ldots,x_n$, decide whether for all assignments to $x_1,\ldots,x_{\lceil \alpha\cdot n \rceil}$ there exists an assignment to the rest of the variables that satisfies $\phi$, that is, whether: 
\[
\forall x_1,\ldots,x_{\lceil \alpha\cdot n \rceil} \exists x_{\lceil \alpha\cdot n \rceil+1},\ldots, x_n: \phi(x_1,\ldots,x_n) = \text{true}. 
\]
\end{hypothesis}

Our lower bounds are provided via a reduction from the AND Subset Sum problem~\cite{AbboudBHS20}. Recall that in Subset Sum, we are given a set of integers $X$ and a target $t$, where each integer $x\in X$ is in the range $0 < x \leq t$. The goal is to determine whether there is a subset $Y \subseteq X$ whose elements sum up to exactly~$t$; \emph{i.e.} $\sum_{x \in Y} x = t$. If this is in fact the case, $(X,t)$ is a \emph{yes-instance}. In the corresponding AND Subset Sum problem, we are given $N$ many Subset Sum instances $(X_1,t_1),\ldots,(X_N,t_N)$, and the goal is to determine whether all instances are yes-instances. We have the following relationship between $\forall \exists$-SETH (Hypothesis~\ref{hyp:FESETH}) and the AND Subset Sum problem. 
\begin{theorem}[\cite{AbboudBHS20}]
\label{thm:ANDSubsetSum}%
Assuming $\forall \exists$-SETH, there are no $\delta,\varepsilon > 0$ such that the following problem can be solved in $\widetilde{O}(N^{1+\delta-\varepsilon})$ time: Given $N$ Subset Sum instances, each with $O(N^\varepsilon)$ integers and target $O(N^\delta)$, determine whether all of these instances are yes-instances.
\end{theorem}

Theorem~\ref{thm:ANDSubsetSum} above was used in~\cite{AbboudBHS20} to exclude a running time of $\widetilde{O}(Ns +t(Ns)^{1- \varepsilon})$ for AND Subset Sum and the related AND Partition problem. Here $N$ denotes the number of instances, $s$ is the maximal number of integers in each instance, and $t$ is the maximum target of all instances. For our purposes, we will need a slightly different bound that is given below:
\begin{corollary}
\label{cor:LowerBounds}%
Assuming $\forall \exists$-SETH, there are no $\varepsilon > 0$ and $c > 0$ such that the following problem can be solved in $\widetilde{O}((Ns)^c+(tN)^{1-\varepsilon})$ time: Given $N$ Subset Sum instances, each with at most $s$ integers and target at most $t$, determine whether all of these instances are yes-instances.
\end{corollary} 

\begin{proof}
Assume that we can solve AND Subset Sum in $\widetilde{O}((Ns)^c+(tN)^{1-\varepsilon_0})$ time, for some $c > 0$ and $0 < \varepsilon_0 < 1$. Let $(X,t_1),\ldots,(X_n,t_n)$ be $N$ Subset Sum instances with $|X_i|=O(N^\varepsilon)$ and $t_i=O(N^\delta)$ for each $1 \leq i \leq N$, for some $\delta,\varepsilon > 0$ to be specified later. On these instances, this algorithm would run in time
$$
\widetilde{O}((Ns)^c+(tN)^{1-\varepsilon}) = \widetilde{O}(N^{c(1+\varepsilon)}+N^{(1-\varepsilon_0)(1+\delta)}) = \widetilde{O}(N^{c(1+\varepsilon)}+N^{1+\delta-\varepsilon_0}).
$$
We choose $\delta$ so that $\delta > c$. This allows as to choose an $\varepsilon > 0$ so that $\varepsilon = \min\{(\delta-c)/c,\varepsilon_0\}$. Thus, the running time above can be bounded by
$$
\widetilde{O}(N^{c(1+\varepsilon)}+N^{(1-\varepsilon_0)(1+\delta)}) = \widetilde{O}(N^{\delta}+N^{1+\delta-\varepsilon}) = \widetilde{O}(N^{1+\delta-\varepsilon}),
$$
violating $\forall \exists$-SETH according to Theorem~\ref{thm:ANDSubsetSum}.
\end{proof}

We next briefly describe the reduction in~\cite{AbboudBHS20} from AND Subset Sum to \Prob. The first thing to note is that while the reduction there is from AND Partition, the same construction also holds for AND Subset Sum. Now, the main idea in the construction is that we can cram $N$ Subset Sum instances into a single \Prob instance by using $d_{\#}=N$ many due dates, and by not increasing $d_{\max}$ too much. More precisely, we have the following:
%\begin{lemma}
%\label{lem:LowerBound}%
%There is an algorithm that takes as input $N$ Subset Sum instances $(X_1,t_1),\ldots,$ $(X_N,t_N)$, where there exist $\delta,\varepsilon>0$ such that $|X_i|= O(N^\varepsilon)$ and $t_i= O(N^\delta)$ for all $1\le i\le N$, and outputs in $O(N^{1+\varepsilon})$ time a \Prob instance $J$ such that 
%\begin{itemize}
%\item[$(i)$] $J$ is a yes-instance of \Prob iff each $(X_i,t_i)$ is a yes-instances of Subset Sum.
%\item[$(ii)$] $n = O(N^{1+\varepsilon})$ and $d_{\max} = O(N^{1+\delta})$.
%\end{itemize}
%\end{lemma}
\begin{lemma}[\cite{AbboudBHS20}]
\label{lem:LowerBound}%
There is an algorithm that takes as input $N$ Subset Sum instances $(X_1,t_1),\ldots,$ $(X_N,t_N)$, with $s=\max_i |X_i|$ and $t=\max_i t_i$, and outputs in $O(Nn)$ time a \Prob instance $J$ with $n$ jobs such that 
\begin{itemize}
\item[$(i)$] $J$ is a yes-instance of \Prob iff each $(X_i,t_i)$ is a yes-instances of Subset Sum.
\item[$(ii)$] $n = O(Ns)$ and $d_{\max} = O(Nt)$.
\end{itemize}
\end{lemma}

\begin{proof}[Proof of Theorem~\ref{thm:LowerBound}]
Suppose \Prob has an $\widetilde{O}(n^{O(1)} + d_{\max}^{1-\varepsilon})$ time algorithm for some $\varepsilon > 0$. We show that this violates $\forall \exists$-SETH, by proving a fast algorithm for AND Subset Sum. Let $(X_1,t_1),\ldots,$ $(X_N,t_N)$ be an input to AND Subset Sum, and let $s=\max_i |X_i|$ and $t=\max_i t_i$. Apply Lemma~\ref{lem:LowerBound} to obtain an equivalent \Prob instance $J$ with $n = O(Ns)$ jobs and maximum due date $d_{\max} = O(Nt)$. Applying the assumed \Prob algorithm on this instance solves the AND Subset Sum instance in $\widetilde{O}(n^{O(1)} + d_{\max}^{1-\varepsilon_0}) = \widetilde{O}((Ns)^{O(1)} + (Nt)^{1-\varepsilon})$ time, which is impossible assuming $\forall \exists$ SETH according to Corollary~\ref{cor:LowerBounds}.
\end{proof}

\section{Conclusion}

We identified new scenarios of \Prob where it is possible to improve upon Lawler and Moore's classic algorithm~\cite{LawlerMoore}. Our algorithm is based on $(\max,+)$-convolutions, and can be sped up using the recent improvements on special cases of $(\max,+)$-computations. Using only parameters $n$ and $d_{\max}$, our algorithm can be bounded by $\widetilde{O}(n+d^3_{\max})$. Thus, recalling the $\widetilde{O}((n+d^2_{\max})^{2-\varepsilon})$ of Cygan~\emph{et al.}, the obvious most important open question is the following: 
\begin{quote}
``Can \Prob be solved in $\widetilde{O}(n+d^2_{\max})$ time, or is does it require $\widetilde{O}(n+d^3_{\max})$ time? Or perhaps something in between?"
\end{quote}
Note that even an $\widetilde{O}(n^{O(1)}+d^2_{\max})$ time algorithm would be desirable. As a step towards solving this question, one can also attempt at improving the dependency on $d_{\#}$ in Theorem~\ref{thm:MainAlgorithm}, perhaps to something like $\widetilde{O}(n^{O(1)}+\sqrt{d_{\#}}d^2_{\max})$ time. 

Regarding parameters $p_{\max}$ and $w_{\max}$, currently the fastest known Knapsack algorithm with respect to these parameters alone is the $\widetilde{O}(n+\min\{p^3_{\max},w^3_{\max}\})$ time algorithm by Polak \emph{et al.}~\cite{PolakEtAl21}. Note that while their result is obtained also via $(\max,+)$-convolutions (amongst other things), it only computes a single optimal value, and not an entire solution array as all other Knapsack algorithms used in this paper. Thus, it seems challenging to adapt their techniques to \Prob. Nevertheless, any algorithm with running time of the form $\widetilde{O}(n+p^{O(1)}_{\max})$ or $\widetilde{O}(n+p^{O(1)}_{\max})$ would be interesting. Furthermore, an $\widetilde{O}(d_{\#}n+d_{\#}d_{\max}w_{\max})$ time algorithm to match the $\widetilde{O}(d_{\#}n+d_{\#}d_{\max}p_{\max})$ time algorithm of Theorem~\ref{thm:SecondSpeedup} is also desirable.

Finally, it would be very interesting to see what other scheduling problems can benefit from fast $(\max,+)$-convolutions. We believe this tool should prove useful in other unrelated scheduling settings.

\bibliographystyle{plain}
\bibliography{biblo}

\end{document}